\documentclass[a4paper,12pt]{article}
\usepackage[utf8x]{inputenc}
\usepackage{amssymb,amsmath,amsthm,amstext,graphics,amsfonts,hyperref}
\newtheorem{theorem}{Theorem}[]
\newtheorem{lemma}[theorem]{Lemma}

\newtheorem{corollary}[theorem]{Corollary}
\newtheorem{proposition}[theorem]{Proposition}
\theoremstyle{definition}
\newtheorem{example}[theorem]{Example}
\newtheorem{definition}[theorem]{Definition}

\newtheorem{alg}[theorem]{Algorithm}

\DeclareMathOperator{\lcm}{lcm}

\title{Permutation Codes over Finite Fields}
\author{Irwansyah\\
	{\it \small Department of Mathematics,}\\
	{\it \small Faculty of Mathematics and Natural Sciences,}\\
	{\it \small Universitas Mataram, Jl. Majapahit 62, Mataram, 83125}\\
	{\it \small INDONESIA}\\
	{\it \small Email: irw@unram.ac.id}\\
	\vspace{0.1cm}\\
	Intan Muchtadi-Alamsyah and Aleams Barra\\
	{\it \small Algebra Research Group,}\\
	{\it \small Faculty of Mathematics and Natural Sciences,}\\
	{\it \small Institut Teknologi Bandung,
		Jl. Ganesha 10, Bandung, 40132,}\\
	{\it \small INDONESIA}\\
	{\it \small Email: ntan@math.itb.ac.id, barra@math.itb.ac.id}}
\date{}

\begin{document}
	\maketitle
	
	\begin{abstract}
		In this paper we describe a class of codes called {\it permutation codes}. This class of codes is a generalization of cyclic codes and quasi-cyclic codes. We also give some examples of optimal permutation codes over binary, ternary, and $5$-ary. Then, we describe its structure as submodules over a polynomial ring.
		\\[0.25cm]
	\end{abstract}

\section{Introduction}

Cyclic code is one important type of codes. This type of codes over finite field $\mathbb{F}_q$ can be considered as ideals in quotient ring $\displaystyle{\frac{\mathbb{F}_q[x]}{\langle x^n-1\rangle}},$ where $n$ is the length of codes. Based on this point of view, we can determine generator of any cyclic code, its Euclidean dual, and its dimension. Moreover, in some cases, we can also design the minimum distance and formulate decoding algorithm for cyclic codes. For more details, see \cite{pless}.

The other important type of codes is quasi-cyclic code. This type of codes is a generalization of cyclic code. Quasi-cyclic codes can be viewed as modules over a finite polynomial ring, and decomposed by the Chinese Remainder Theorem or discrete Fourier transform into products of shorter codes over larger alphabets. Based on this point of view, we can construct self-dual quasi-cyclic codes explicitly, derive a new quartenary construction of Leech lattice, enumerate self-dual one generator quasi-cyclic codes, and formulate some constructions for codes such as squaring, cubing, quinting, and septing constructions. See \cite{ling,ling2}. Cyclic and quasi-cyclic codes have several applications such as images transmision from mars to earth, compact disk storage, and being used as public keys with compact structure for McElice's cryptosystem. \\[0.25cm]

In this paper, we describe a class of codes called {\it permutation codes}. This class of codes is a generalization of cyclic codes and quasi-cyclic codes. We describe its algebraic structure and give some examples of optimal permutation codes over binary, ternary, and $5$-ary. 

\section{Basic Facts}\label{basic}

Let $C$ be a code of length $n$ over finite field $\mathbb{F}_q,$ where $q=p^r,$ for some prime number $p$ and natural number $r.$ Also, let $S_n$ be the permutation group on $n$ elements. Now, we define a class of codes as follow.

\begin{definition}
	A code $C$ is said to be a {\it permutation code} or {\it $\sigma$-code}, for some $\sigma\in S_n,$ if  for any $\textbf{c}$ in $C,$ we have
	\begin{equation}
	T_{\sigma}(\textbf{c})=\left(c_{\sigma^{-1}(1)},
	c_{\sigma^{-1}(2)},\dots,c_{\sigma^{-1}(n)}\right)
	\end{equation}
	is also in $C.$
\end{definition}

Note that, a permutation code is a code which is globally invariant under the action of a given permutation group as in \cite[Chapter 17]{pless}. Here are some examples of permutation codes.

\begin{enumerate}
	\item {\bf Cyclic Code.} A cyclic code can be considered as a $\sigma$-code, where $\sigma=(1\;2\;\cdots\;n)\in S_n.$
	
	\item {\bf Quasi-Cyclic Code.} A quasi-cyclic code is a $\sigma$-code, where $\sigma=(1\;1+d\;1+2d\cdots\;1+(l-1)d)(2\;2+d\;\cdots\;2+(l-1)d)\cdots(d-1\;d-1+d\cdots\;d-1+(l-1)d)\in S_n.$

\label{example1}
\end{enumerate}

For any code $C,$ let $C^\bot$ be the Euclidean dual of $C.$ The following proposition shows that the dual of a permutation code is also a permutation code.

\begin{proposition}
	If $C$ is a $\sigma$-code, then $C^\bot$ is also a $\sigma$-code.
	\label{dual}
\end{proposition}

\begin{proof}
	Let $\textbf{c}'=(c_1',\dots,c_n')$ be any element in $C^\bot.$ We need to show that $T_{\sigma}(\textbf{c}')$ is also in $C^\bot.$ For any $\textbf{c}''$ in $C,$ there exists $\textbf{c}$ in $C$ such that $T_{\sigma}(\textbf{c})=\textbf{c}''$ because $C$ is a $\sigma$-code. Now, consider
	\[\langle \textbf{c}'',T_{\sigma}(\textbf{c}')\rangle=\langle T_{\sigma}(\textbf{c}),T_{\sigma}(\textbf{c}')\rangle=
	\sum_{i=1}^nc_{\sigma^{-1}(i)}c_{\sigma^{-1}(i)}'=0.\]
	This gives $T_{\sigma}(\textbf{c}')\in C^\bot$ as we hope.
\end{proof}

Let $R=\mathbb{F}_q[Y]/\langle Y^n-1\rangle$ and define a left action of $\mathbb{F}_q[Y]$ on $R$ as follows.
For any $a\in R,$ let $a=f(Y)+\langle Y^n-1\rangle,$ and for any
$h(Y)\in \mathbb{F}_q[Y;\theta],$ we define
\[
h(Y)*a=h(Y)*f(Y)+\langle Y^n-1\rangle
\]
we can show that this left action is well-defined and $R$ is a left module over $\mathbb{F}_q[Y].$

Let $\sigma=\sigma_1\sigma_2\cdots \sigma_k,$ where $\sigma_i = (t_i\;\sigma(t_i)\;\cdots\;\sigma^{m_i-1}(t_i))$ is a cycle of length $m_i,$ for some $t_i$ in $\mathbb{N},$ for all $i=1,2,\dots,k.$ Also, let $\displaystyle{R_i=\frac{\mathbb{F}_q[Y]}{\langle Y^{m_i}-1\rangle}},$ for all $i=1,2,\dots,k.$ Define a map,

\begin{equation}
\begin{array}{llll}
\varphi : & \mathbb{F}_q^n & \longrightarrow & R_1\times R_2\times\cdots\times R_k\\
& \mathbf{c}=(c_1,c_2,\dots,c_n) & \longmapsto & \left(\mathbf{c}_1(Y),\mathbf{c}_2(Y),\dots,\mathbf{c}_k(Y)\right),
\end{array}
\end{equation}
where $\mathbf{c}_i(Y)=\sum_{j=0}^{m_i-1}c_{\sigma^{j}(t_i)}Y^j,$ for all $i=1,2,\dots,k.$ Let $\varphi(C)$ be the image of $C$ under the map $\varphi.$ We have the following proposition.

\begin{proposition}\label{phi}
	The map $\varphi$ induces a {\it one-to-one} correspondence between $\sigma$-codes of length $n$ over $\mathbb{F}_q$ and submodules of $R_1\times R_2\times \cdots\times R_k$ over $\mathbb{F}_q[Y].$
\end{proposition}

\begin{proof}
	Let $C$ be a $\sigma$-code of length $n$ over $\mathbb{F}_q.$ The $\varphi(C)$ will be closed under the multiplication by elements of $\mathbb{F}_q$ because $C$ is a linear code. Since $Y^{m_i}=1$ in $R_i,$ for all $i=1,2,\dots,k,$ consider
	\[Y\mathbf{c}_i(Y)=\sum_{j=0}^{m_i-1}c_{\sigma^j(t_i)}Y^{j+1}=c_{\sigma^{m_i-1}(t_i)}+c_{t_i}Y+c_{\sigma(t_i)}Y^2
+\cdots+c_{\sigma^{m_i-2}(t_i)}Y^{m_i-1}.\]
	The above equation implies, for any $\mathbf{c}=(c_1,\dots,c_n)$ in $\mathbb{F}_q^n,$
	\[\varphi\left(T_{\sigma}(\mathbf{c})\right)=\left(Y\mathbf{c}_1(Y),Y\mathbf{c}_2(Y),\dots,Y\mathbf{c}_1(Y)\right).\]
	So, $\varphi(C)$ also closed under the multiplication by $Y$ and the action $T_{\sigma}$ in $C$ is correspond to the multiplication by $Y$ in $R_1\times\cdots\times R_k.$ Therefore, $\varphi(C)$ is a submodule of $R_1\times\cdots\times R_k$ over $\mathbb{F}_q[Y].$
\end{proof}

\section{Good Permutation Codes}

The results in the previous section give us a simple systematic way to construct permutation codes. Therefore, in this part, we will construct permutation codes using \verb|Octave|. Due to the limited memory in \verb|Octave|, we only construct codes with small length and dimension.

Here is an example of $\sigma$-code obtained using the corresponding submodule as in Proposition~\ref{phi}.

\begin{example}\label{contohsatu}
	Let $\sigma = (1\;2\;3)(4\;5).$ We would like to find $\sigma$-code of length $5$ over $\mathbb{F}_2.$ Consider a map
	\[\begin{array}{llll}
	\varphi : & \mathbb{F}_2^5 & \longrightarrow & \displaystyle{\frac{\mathbb{F}_2[Y]}{\langle Y^3-1\rangle}}\times \displaystyle{\frac{\mathbb{F}_2[Y]}{\langle Y^2-1\rangle}}\\
	& (c_1,c_2,c_3,c_4,c_5) & \longmapsto & (c_1+c_2Y+c_3Y^2,c_4+c_5Y).
	\end{array} \]
	Now, choose $C=\langle (1+Y,1+Y)\rangle\subseteq \displaystyle{\frac{\mathbb{F}_2[Y]}{\langle Y^3-1\rangle}}\times \displaystyle{\frac{\mathbb{F}_2[Y]}{\langle Y^2-1\rangle}}.$ We can see that
	\[Y(1+Y,1+Y)=(Y+Y^2,1+Y)\]
	and
	\[Y^2(1+Y,1+Y)=(1+Y^2,1+Y).\]
	 So, we have
	\[C=\langle(1+Y,1+Y),(Y+Y^2,1+Y),(1+Y^2,1+Y)\rangle.\]
	This means, $\varphi^{-1}(C)=\langle (0,1,1,1,1),(0,1,1,1,1),(1,0,1,1,1)\rangle.$ The code $\varphi^{-1}(C)$ is a $\sigma$ binary code with dimension 3 and Hamming distance 2.
\end{example}

We use the following simple algorithm, based on the result in previous section, to construct permutation codes of length $n.$

\begin{alg}\label{construction}
Let $T$ be a shift operator such that $T(a_1,a_2,\dots,a_m)=(a_m,a_1,a_2,\dots,a_{m-1}).$
\begin{itemize}
\item[1.] Choose a permutation $\sigma\in S_n,$ where $\sigma=\sigma_1\cdots\sigma_t,$ and $\sigma_i$ is a cycle of length $m_i,$ such that $lcm(m_1,\dots,m_t)=k.$

\item[2.] Choose a vector $\mathbf{a}=\left(\mathbf{a}^{1}|\mathbf{a}^{2}|\cdots|\mathbf{a}^{t}\right),$ where $\mathbf{a}^{i}=(a_{i1},a_{i2},\dots,a_{im_i})\in\mathbb{F}^{m_i},$ such that vectors $\mathbf{a}^{i},T(\mathbf{a}^{1}),\dots,T^{m_i-1}(\mathbf{a}^{1})$ are linearly independent, for all $i=1,2,\dots,t.$

\item[3.] Generate vectors $\mathbf{a},T_{\sigma}(\mathbf{a}),\dots,T_{\sigma}^{k-1}(\mathbf{a}).$

\item[4.] Generate $\sigma$-code $C$ with generators $\mathbf{a},T_{\sigma}(\mathbf{a}),\dots,T_{\sigma}^{k-1}(\mathbf{a}).$

\end{itemize}
\end{alg}

Using a similar way as in Example~\ref{contohsatu} and Algorithm~\ref{construction}, we construct some optimal binary, ternary, and 5-ary $\sigma$-codes as shown in Tables~\ref{binary},\ref{ternary}, and \ref{5ary}, where the optimality is based on tables of optimal linear codes in \url{www.codetables.de}.
Note that, generator given in the table is for the corresponding submodule. The letters $k$ and $d$ are notations for dimension and Hamming distance of the corresponding binary/ternary/5-ary code, respectively.  \\[0.25cm]

\section{Algebraic Structure of Permutation Codes}

\subsection{Permutation codes as torsion submodules}
Let $\sigma = \sigma_1\sigma_2\cdots\sigma_k$ in $S_n,$ where $\sigma_i$ is a cycle of length $m_i,$ for all $i=1,2,\dots,k.$ As we already show in Section~\ref{basic}, a $\sigma$-code can be considered as a submodule of $M=M(q,m_1,\dots,m_k)=R_1\times R_2\times \cdots\times R_k$ over $\mathbb{F}_q[Y],$ where $R_i=\mathbb{F}_q[Y]/\langle Y^{m_i}-1\rangle,$ for all $i=1,2,\dots,k.$ In this section, we will describe algebraic structure of permutation codes by viewing $M$ as a torsion module over $\mathbb{F}_q[Y].$ Here we recall the definition of torsion module.

\begin{definition}{\cite{roman}}
 Let $N$ be a module over a ring $R.$
 \begin{itemize}
 \item A non-zero element $v$ in $N$ for which $rv=0$ for some non-zero element $r$ in $R,$ is called a {\bf torsion element}.

 \item If all elements of $N$ are {\bf torsion elements}, then $N$ is called a {\bf torsion module}.
 \end{itemize}
\end{definition}

The following proposition shows that $M=M(q,m_1,\dots,m_k)=R_1\times R_2\times \cdots\times R_k$ is a torsion module over $\mathbb{F}_q[Y].$

\begin{proposition}
The module $M$ is a torsion module over $\mathbb{F}_q[Y].$ Moreover, the order of $M,$ $o(M),$ is equal to $lcm\left(Y^{m_1}-1,\dots,Y^{m_k}-1\right).$
\end{proposition}

\begin{proof}
If $r(Y)=lcm\left(Y^{m_1}-1,\dots,Y^{m_k}-1\right),$ then $ra=0,$ for all $a$ in $M.$ So, $M$ is a torsion module over $\mathbb{F}_q[Y].$  Moreover, if $S$ the annihilator ideal of $M,$ then $r(Y)\in S.$ Since, $\mathbb{F}_q[Y]$ is a principal ideal domain, assume that $S=\langle g(Y)\rangle,$ for some $g(Y)$ in $\mathbb{F}_q[Y].$ Suppose that $\deg(g)<\deg(r),$ then there exists $i\in \{1,2,\dots,k\}$ such that $g\not\equiv 0 \mod (y^{m_i}-1).$ Consequently, if we choose $a=(0,\dots,0,1,0,\dots,0)\in M,$ then $ga\not=0,$ a contradiction. So, $r=bg,$ for some $b\in\mathbb{F}_q^\times.$ Therefore, $S=\langle r(Y)\rangle,$ or $o(M)=r(Y)=lcm\left(Y^{m_1}-1,\dots,Y^{m_k}-1\right).$
\end{proof}

Let $lcm(Y^{m_1}-1,\dots,Y^{m_k}-1)=\prod_{j=1}^tf_j(Y)^{\alpha_j},$ for some irreducible polynomial $f_j(Y)$ and integer $\alpha_j\geq 1,$ for all $j=1,2,\dots,t.$ Then, based on the primary decomposition theorem \cite[Theorem 6.10]{roman}, we have

\begin{equation}\label{decomposition}
M=\bigoplus_{i=1}^t M_i,
\end{equation}
where $M_i$ is a primary module of order $f_i^{\alpha_i},$ {\it i.e.} $M_i=\{a\in M|f_i^{\alpha_i}a=0\}.$ Moreover, by cyclic decomposition theorem for a primary module \cite[Theorem 6.12]{roman}, we can decompose each $M_i$ as follows.

\begin{equation}\label{decomposition2}
M_i = \bigoplus_{j=1}^{t_i} \langle v_{ir}\rangle,
\end{equation}
with $annihilator(\langle v_{ir}\rangle)=\langle f_i^{e_{ij}}\rangle,$ where $e_{i1}=\alpha_i\geq e_{i2}\geq e_{i3}\geq \cdots\geq e_{it_i}.$
Therefore, we have

\begin{equation}\label{decomposition3}
M=\bigoplus_{i=1}^t\bigoplus_{j=1}^{t_i}\langle v_{ij}\rangle,
\end{equation}
where $o(v_{ij})=f_i^{e_{ij}}$ as in the previous decomposition. Let $\mathcal{R}=\{1,2,\dots,t\},$ we have the following result for permutation codes.

\begin{theorem}\label{permcode}
Let $C$ be a $\sigma$-code over $\mathbb{F}_q,$ and $\mathcal{R}_C\subseteq \mathcal{R},$ where for any $i\in\mathcal{R}_C,$ $f_i|o(C).$ Then,
\begin{itemize}
\item[(a)] The order of $C,$ {\it i.e.} $o(C),$ is equal to $\prod_{i\in\mathcal{R}_C}f_i^{\beta_i},$ where $\beta_i\leq\alpha_i,$ for all $i\in \mathcal{R}_C.$

\item[(b)] The code $C$ can be written as
\[C=\bigoplus_{i\in\mathcal{R}_C}\bigoplus_{j=1}^{t_i}\langle w_{ij}\rangle,\]
for some $w_{ij}\in M$ for which $o(w_{ij})=f_{i}^{e_{ij}},$ where $e_{i1}=\beta_i\geq e_{i2}\geq \cdots\geq e_{it_i}.$

\item[(c)] The dimension of $C$ over $\mathbb{F}_q$ is equal to $deg(o(C)).$
\end{itemize}
\end{theorem}

\begin{proof}
(a) If $C$ is a submodule of $M,$ then $ann(M)\geq ann(C).$ This means, the generator of $ann(C)$ divides $\prod_{j=1}^tf_j(Y)^{\alpha_j}$ as we hope. So, if $\mathcal{R}_C\subseteq \mathcal{R},$ where for any $i\in\mathcal{R}_C,$ $f_i|o(C),$ then $o(C)=\prod_{i\in\mathcal{R}_C}f_i^{\beta_i},$ where $\beta_i\leq\alpha_i,$ for all $i\in \mathcal{R}_C.$\\[0.25cm]
(b) Apply \cite[Theorem 6.12]{roman} as in the previous decomposition for $M.$\\[0.25cm]
(c) Let $g(Y)$ be an element in $C$ for which $o(g)=\prod_{i\in\mathcal{R}_C}f_i^{\beta_i},$ and $\deg(o(g))=s.$ Then, over $\mathbb{F}_q,$ the set $\{g(Y),Yg(Y),\dots,Y^{s-1}g(Y)\}$ is a maximal linearly independent set as we hope.
\end{proof}

\subsection{Duality}\label{duality}

In the previous approach, we have a problem in describing dual of a code in the torsion module $M.$ So, in this part, we will describe a way to see duality for permutation codes easily. Recall that, $M=\displaystyle{\frac{\mathbb{F}_q[Y]}{\langle Y^{m_1}-1\rangle}\times\cdots\times\frac{\mathbb{F}_q[Y]}{\langle Y^{m_k}-1\rangle}},$ $f(Y)=\lcm\left(Y^{m_1}-1,\dots,Y^{m_k}-1\right),$ $\deg(f)=m,$ and $m'=\lcm(m_1,\dots,m_k).$ We have the following properties.

\begin{lemma}\label{cm}
	Polynomial $p$ is a common multiple of $Y^{m_1}-1,Y^{m_2}-1,\dots,$ and $Y^{m_k}-1$ if and only if $pb=0,$ for all $b\in M.$
\end{lemma} 	

\begin{proof}
	($\Leftarrow$) When $p$ is a common multiple of $Y^{m_1}-1,Y^{m_2}-1,\dots,$ and $Y^{m_k}-1,$ we have that $p\equiv 0 \mod(Y^{m_i}-1),$ for all $i=1,\dots,k.$ \\
	($\Rightarrow$) If $pb=0$ for all $b\in M,$ then $p(1,1,\dots,1)=0.$ So, we have $p\equiv 0 \mod(Y^{m_i}-1),$ for all $i=1,\dots,k.$ Therefore, $Y^{m_i}-1|p,$ for all $i=1,\dots,k.$
\end{proof}

\begin{proposition}\label{contain}
	Let $\langle f(Y)\rangle$ be an ideal, in $\mathbb{F}_q[Y],$ generated by $f(Y).$ Then, $Y^{m'}-1$ is an element in $\langle f(Y)\rangle$ and, moreover, $\langle Y^{m'}-1\rangle\subseteq \langle f(Y)\rangle.$
\end{proposition}

\begin{proof}
Since $\sigma=\sigma_1\sigma_2\cdots\sigma_k,$ where $\sigma_i$ is a cycle of length $m_i,$ for all $i=1,2,\dots,k,$ we have $order(\sigma)=m'$ and $T^{m'}_\sigma(\mathbf{a})=\mathbf{a}.$ Recall that $T^j_\sigma(\mathbf{a})$ corresponds to $Y^j(\phi(\mathbf{a})).$ So, we have $Y^{m'}(\phi(\mathbf{a}))=\phi(\mathbf{a})$ or $(Y^{m'}-1)\phi(\mathbf{a})=0.$ By Lemma~\ref{cm}, $Y^{m'}-1$ is a common multiple of $Y^{m_1}-1,Y^{m_2}-1,\dots,$ and $Y^{m_k}-1.$ Therefore, $f(Y)|Y^{m'}-1.$ 	
\end{proof}

Based on Proposition~\ref{contain}, it is natural to define an injective map from $\mathbb{F}_q^n$ to $M'=\displaystyle{\frac{\mathbb{F}_q[Y]}{\langle Y^{m'}-1\rangle}\times\cdots\times\frac{\mathbb{F}_q[Y]}{\langle Y^{m'}-1\rangle}}.$
Without loss of generality, assume that $\sigma_i=(1+\sum_{j=1}^{i-1}m_j,2+\sum_{j=1}^{i-1}m_j,\dots,\sum_{j=1}^{i}m_j),$ for all $i=2,\dots,k,$ and $\sigma_1=(1,2,\dots,m_1).$ Any $\mathbf{a}$ in $\mathbb{F}_q^n$ can be written as
\[\mathbf{a}=\left(\mathbf{a}_1|\mathbf{a}_2|\cdots|\mathbf{a}_k\right),\]
where $\mathbf{a}_i\in\mathbb{F}_q^{m_i},$ for all $i=1,2,\dots,k.$
 {\it First}, define a map from $\mathbb{F}_q^n,$ where $n=m_1+m_2+\cdots+m_k,$ to $\mathbb{F}_q^{m'k}$ as follows.

\begin{equation}\label{lambda1}
\begin{array}{llll}
\lambda_1 : & \mathbb{F}_q^n & \longrightarrow & \mathbb{F}_q^{m'k}\\
 & \mathbf{a} & \longmapsto & \left(\mathbf{a}^{(1)}|\mathbf{a}^{(2)}|\cdots|\mathbf{a}^{(k)}\right),
\end{array}
\end{equation}
with
\begin{equation}\label{copy}
\mathbf{a}^{(i)}=\underbrace{\left(\mathbf{a}_i|\mathbf{a}_i|\cdots|\mathbf{a}_i\right)}_{n_i},
\end{equation}
where $n_i=\displaystyle{\frac{m'}{m_i}}.$ {\it Second}, let $\mathbf{a}^{(i)}=\left(a_{i1},a_{i2},\dots,a_{im'}\right),$ and define a map from $\lambda_1\left(\mathbb{F}_q^{n}\right)$ to $\mathbb{F}_q^{m'k}$ as follows.

\begin{equation}\label{lambda2}
\begin{array}{llll}
\lambda_2 : & \lambda_1\left(\mathbb{F}_q^{n}\right) & \longrightarrow & \mathbb{F}_q^{m'k}\\
 & \left(\mathbf{a}^{(1)}|\mathbf{a}^{(2)}|\cdots|\mathbf{a}^{(k)}\right) & \longmapsto & \left(\mathbf{a}_{(1)}|\mathbf{a}_{(2)}|\cdots|\mathbf{a}_{(m')}\right),
\end{array}
\end{equation}
where $\mathbf{a}_{(j)}=(a_{1j},a_{2j},\dots,a_{kj}),$ for all $j=1,2,\dots,m'.$ Now, we shall define a map from $\mathbb{F}_q^n$ to $\mathbb{F}_q^{m'k}$ as follows.

\begin{equation}\label{lambda}
\begin{array}{llll}
\lambda : & \mathbb{F}_q^{n} & \longrightarrow & \mathbb{F}_q^{m'k}\\
& \mathbf{a} & \longmapsto & \lambda_2(\lambda_1(\mathbf{a})).
\end{array}
\end{equation}

We have the following proposition related to the map $\lambda.$

\begin{proposition}
	If $C$ is a $\sigma$-code of length $n,$ then $\lambda(C)$ is a quasi-cyclic code of length $m'k$ with index $k.$
\end{proposition}

\begin{proof}
	We can check that $\lambda\left(T_\sigma(\mathbf{a})\right)=T^k(\lambda(\mathbf{a})).$ Therefore, if $T_\sigma(\mathbf{a})\in C,$ then $T^k(\lambda(\mathbf{a}))\in \lambda(C).$
\end{proof}

{\it Third}, any $\mathbf{b}\in \mathbb{F}_q^{m'k},$ can be written as
\[\mathbf{b}=\left(b_{11},b_{12},\dots,b_{1k},\dots,b_{m'1},b_{m'2},\dots,b_{m'k}\right).\]
Now, define a map from $\mathbb{F}_q^{m'k}$ to $M'=\displaystyle{\frac{\mathbb{F}_q[Y]}{\langle Y^{m'}-1\rangle}\times\cdots\times\frac{\mathbb{F}_q[Y]}{\langle Y^{m'}-1\rangle}}$ as follows.

\begin{equation}\label{phi}
\begin{array}{llll}
\phi : & \mathbb{F}_q^{m'k} & \longrightarrow & M'\\
& \mathbf{b} & \longmapsto & \left(b_1(Y),b_2(Y),\dots,b_k(Y)\right),
\end{array}
\end{equation}
where $b_i(Y)=\sum_{j=0}^{m'-1}b_{(j+1)i}Y^j,$ for all $i=1,2,\dots,k.$ The map $\phi$ is a {\it one-to-one} correspondence between quasi-cyclic codes of length $m'k$ and $\displaystyle{\frac{\mathbb{F}_q[Y]}{\langle Y^{m'}-1\rangle}}$-submodules of $M',$ see \cite{ling,ling2} for more details. By composing $\lambda$ and $\phi,$ we have the following map.

\begin{equation}\label{miu}
\begin{array}{llll}
\mu : & \mathbb{F}_q^{n} & \longrightarrow & M'\\
& \mathbf{a} & \longmapsto & \phi(\lambda(\mathbf{a})).
\end{array}
\end{equation}

For our convenience, we shall define the following notion.
\begin{definition}\label{coef}
	A vector $\mathbf{a}=(a_1,\dots,a_t)$ in $\mathbb{F}_q^t$ is said to be the {\it coefficients vector} for a polynomial $f(Y)$ if $f(Y)=\sum_{i=0}^{t-1}a_{i+1}Y^i.$
\end{definition}

We have the following properties related to the image of $\mu.$

\begin{lemma}\label{imagemiu1}
	If $\mu(\mathbf{a})=\left(a_1(Y),\dots,a_k(Y)\right),$ then $a_i(Y)=f_i(Y)\sum_{j=0}^{n_i-1}Y^{jm_i},$ for some $f_i(Y)\in\displaystyle{\frac{\mathbb{F}_q[Y]}{\langle Y^{m_i}-1\rangle}}$ with coefficients vector $\mathbf{a}_i,$ where $n_i=\displaystyle{\frac{m'}{m_i}},$ for all $i=1,2,\dots,k.$
\end{lemma}

\begin{proof}
	We can check that the coefficients vector for $a_i(Y)$ is $\mathbf{a}^{(i)}.$ By equation~\ref{copy}, we have that
	\[a_i(Y)=f_i(Y)\sum_{j=0}^{n_i-1}Y^{jm_i},\]
	for some $f_i(Y)\in\displaystyle{\frac{\mathbb{F}_q[Y]}{\langle Y^{m_i}-1\rangle}}$ with coefficients vector $\mathbf{a}_i,$ where $n_i=\displaystyle{\frac{m'}{m_i}}.$
\end{proof}

\begin{proposition}\label{imagemiu2}
A code $C$ is a $\sigma$-code of length $n$ over $\mathbb{F}_q$ if and only if $\mu(C)$ is a $\displaystyle{\frac{\mathbb{F}_q[Y]}{\langle Y^{m'}-1\rangle}}$-submodules of $M',$ where for any $\mathbf{c}$ in $\mu(C)$ with $\mathbf{c}=(c_1(Y),\dots,c_{k}(Y)),$ we have
\[c_i(Y)=f_i(Y)\sum_{j=0}^{n_i-1}Y^{m_ij},\]
for some $f_i(Y)\in\displaystyle{\frac{\mathbb{F}_q[Y]}{\langle Y^{m_i}-1\rangle}}$ with coefficients vector $\mathbf{c}_i\in\mathbb{F}_q^{m_i},$ for all $i=1,2,\dots,k.$ 	
\end{proposition}

\begin{proof}
	Apply Lemma~\ref{imagemiu1} and the fact that $\phi(C)$ is a $\displaystyle{\frac{\mathbb{F}_q[Y]}{\langle Y^{m'}-1\rangle}}$-submodules of $M'.$
\end{proof}

Before we describe duality in $M',$ we need to show the following property.

\begin{proposition}\label{duallambda}
Let $C$ be a code of length $n$ over $\mathbb{F}_q$ and $C^\bot$ be its Euclidean dual. If $C_1=\{\mathbf{c}'\in\lambda(\mathbb{F}_q^{n})\left|\right. \mathbf{c}'\cdot\lambda(\mathbf{c})=0,\forall \mathbf{c}\in C\},$ then $C_1=\lambda(C^\bot).$
\end{proposition}

\begin{proof}
We can see that $\lambda(C^\bot)\subseteq C_1.$ Also, we have $dim(C_1)=n-dim(\lambda(C))=n-dim(C)=dim(C^\bot).$ Therefore, $C_1=\lambda(C^\bot).$
\end{proof}
Proposition~\ref{duallambda} shows that any $\mathbf{c}'$ in $\lambda(\mathbb{F}_q^n),$ which satisfies $\mathbf{c}'\cdot\lambda(\mathbf{c})=0,$ for all $\mathbf{c}$ in $C,$ then $\mathbf{c}'\in \lambda(C^\bot).$

Now, define a conjugation map, denoted by $^-,$ on $\displaystyle{\frac{\mathbb{F}_q[Y]}{\langle Y^{m'}-1\rangle}},$ where $\overline{\alpha}=\alpha,$ for all $\alpha$ in $\mathbb{F}_q,$ and $\overline{Y}=Y^{m'-1}.$ Also, define Hermitian inner product on $M'$ as follows: for $\mathbf{a}=(a_1,\dots,a_k)$ and $\mathbf{b}=(b_1,\dots,b_k)$ in $M',$
\[\langle\mathbf{a},\mathbf{b}\rangle=\sum_{i=1}^ka_i\overline{b_i}.\]

We have the following proposition.
\begin{proposition}\label{product}
	Let $\mathbf{a},\mathbf{b}\in\mathbb{F}_q^n.$ Then, $T_{\sigma}^{j}(\mathbf{a})\cdot \mathbf{b}=0,$ for all $0\leq j\leq \omega-1,$ if and only if $\langle\mu(\mathbf{a}),\mu(\mathbf{b}) \rangle=0.$
\end{proposition}

\begin{proof}
	We can see that, $T^{\alpha k}(\lambda(\mathbf{a}))\cdot\lambda(\mathbf{b})=0,$ for all $0\leq \alpha\leq m'-1$ if and only if $T_\sigma^j(\mathbf{a})\cdot\mathbf{b}=0,$ for all $0\leq j\leq m'-1.$ By \cite[Proposition 3.2]{ling2}, $T^{\alpha k}(\lambda(\mathbf{a}))\cdot\lambda(\mathbf{b})=0,$ for all $0\leq \alpha\leq m'-1$ if and only if $\langle \mu(\mathbf{a}),\mu(\mathbf{b})\rangle=0,$ as we hope.
\end{proof}

As a consequence, we have the following result.

\begin{corollary}
	If $C$ be an $\sigma$-code of length $n$ over $\mathbb{F}_q,$ $\varphi(C)$ is its image under the map $\varphi,$ and
\[C_2=\{\mathbf{c}'\in\mu(\mathbb{F}_q^n)\left|\right. \langle\mu(\mathbf{c}),\mathbf{c}'\rangle=0,\forall\mathbf{c}\in C\},\]
 then
	\begin{itemize}
		\item[(i)] The equation $\mu\left(C^\bot\right)=C_2$ holds, and
		\item[(ii)]  the code $C$ is Euclidean self-dual over $\mathbb{F}_q$ if and only if the code $\mu(C)$ is Hermitian self-dual over $\mathbb{F}_q[Y]$ in $\mu(\mathbb{F}_q^n).$
	\end{itemize}
\end{corollary}

\begin{proof}
	Apply Proposition~\ref{duallambda} and \cite[Corollary 3.3]{ling2}.
\end{proof}
Recall that, by Proposition~\ref{imagemiu2}, a $\sigma$-code $C$ is Euclidean self-dual if and only if $\mu(C)$ is Hermitian self-dual over $\mathbb{F}_q[Y]$ in $\mu(\mathbb{F}_q^n),$ where for any $c(Y)\in \mu(C),$ with $\mathbf{c}=(c_1(Y),\dots,c_{k}(Y)),$ we have
\[c_i(Y)=f_i(Y)\sum_{j=0}^{n_i-1}Y^{m_ij},\]
for some $f_i(Y)\in\displaystyle{\frac{\mathbb{F}_q[Y]}{\langle Y^{m_i}-1\rangle}}$ with coefficients vector $\mathbf{c}_i\in\mathbb{F}_q^{m_i},$ for all $i=1,2,\dots,k.$

\subsection{More on Algebraic Structure}

In Subsection~\ref{duality}, we show that we can 'put' $\sigma$-codes of length $n$ over $\mathbb{F}_q,$ where $n=\sum_{i=1}^km_i,$ inside quasi-cyclic codes of length $m'=\lcm(m_1,\dots,m_k)$ over $\mathbb{F}_q.$ Specifically, any $\sigma$-code of length $n$ over $\mathbb{F}_q$ can be considered as a submodule of $\left(\displaystyle{\frac{\mathbb{F}_q[Y]}{\langle Y^{m'}-1\rangle}}\right)^k$ over $\displaystyle{\frac{\mathbb{F}_q[Y]}{\langle Y^{m'}-1\rangle}}$ with some additional conditions in its coordinates. In this part, we will describe more explicit form for these specific submodules. \\[0.3cm]

Let $q=p^r,$ for some prime number $p$ and positive integer $r\geq 1.$ Also, let $m'=p^a\dot{m},$ where $gcd(p,\dot{m})=1.$ The polynomial $Y^{\dot{m}}-1$ factors completely into distinct irreducible factors in $\mathbb{F}_q[Y]$ as follows.

\begin{equation}\label{factor}
Y^{\dot{m}}-1=\delta g_1\dots g_sh_1h_1^*\dots h_th_t^*,
\end{equation}
where $\delta\in\mathbb{F}_q^\times,$ $g_1,\dots,g_s$ are self-reciprocal factors, and $h_j,h_j^*$ are reciprocal pair for all $j=1,2,\dots,t.$ Now, we have

\begin{equation}\label{factor2}
Y^{m'}-1=Y^{p^a\dot{m}}-1=\delta^{p^a} g_1^{p^a}\dots g_s^{p^a}h_1^{p^a}(h_1^*)^{p^a}\dots h_t^{p^a}(h_t^*)^{p^a}.
\end{equation}

As a consequence, we have

\begin{equation}\label{decomp}
\displaystyle{\frac{\mathbb{F}_q[Y]}{\langle Y^{m'}-1\rangle}}=\left(\bigoplus_{i=1}^sG_i\right)\oplus\left(\bigoplus_{j=1}^tH_j\oplus H_j'\right),
\end{equation}
where $G_i=\displaystyle{\frac{\mathbb{F}_q[Y]}{\langle g_i^{p^a}\rangle}},$ for all $i=1,2,\dots,s,$ $H_j=\displaystyle{\frac{\mathbb{F}_q[Y]}{\langle h_j^{p^a}\rangle}}$ and $H_j'=\displaystyle{\frac{\mathbb{F}_q[Y]}{\langle (h_j^*)^{p^a}\rangle}},$ for all $j=1,2,\dots,t.$ So, from \ref{decomp}, we have

\begin{equation}\label{decomp2}
\left(\displaystyle{\frac{\mathbb{F}_q[Y]}{\langle Y^{m'}-1\rangle}}\right)^k= \left(\bigoplus_{i=1}^sG_i^k\right)\oplus\left(\bigoplus_{j=1}^tH_j^k\oplus (H_j')^k\right).
\end{equation}
Therefore, any submodule $C$ of $\left(\displaystyle{\frac{\mathbb{F}_q[Y]}{\langle Y^{m'}-1\rangle}}\right)^k$ over $\displaystyle{\frac{\mathbb{F}_q[Y]}{\langle Y^{m'}-1\rangle}}$ can be decomposed as

\begin{equation}\label{decomp3}
C=\left(\bigoplus_{i=1}^sC_i\right)\oplus\left(\bigoplus_{j=1}^t(C_j'\oplus C_j'')\right),
\end{equation}
where $C_i$ is a submodule of $G_i^k$ over $G_i,$ for all $i=1,2,\dots,s,$ $C_j'$ and $C_j''$ are submodules of $H_j^k$ over $H_j$ and $(H_j')^k$ over $H_j',$ respectively, for all $j=1,2,\dots,t.$

Now, let $\mathbf{b}_i=\sum_{j=0}^{n_i-1}Y^{jm_i},$ where $n_i=\displaystyle{\frac{m'}{m_i}},$ and $\mathbf{b}_{if}=\mathbf{b}_i\mod f(Y).$ Therefore, we have the following results.

\begin{theorem}\label{sigmadecomp}
A code $C$ is a $\sigma$-code of length $n$ over $\mathbb{F}_q$ if and only if
\[\mu(C)=\left(\bigoplus_{i=1}^sC_i\right)\oplus\left(\bigoplus_{j=1}^t(C_j'\oplus C_j'')\right),\]
where
\begin{itemize}
\item[(i)] For any $\mathbf{c}_i\in C_i\leq G_i^k,$ $\mathbf{c}_i=(c_{i1},\dots,c_{ik}),$ where $c_{il}=f_l\mathbf{b}_{lg_i},$ for some $f_l\in G_i,$ for all $l=1,2,\dots,k$ and $i=1,2,\dots,s,$ and

\item[(ii)] For any $\mathbf{c}_j'\in C_j'\leq H_j^k$ and $\mathbf{c}_j''\in C_j''\leq (H_j')^k,$  $\mathbf{c}_j'=(c_{j1}',\dots,c_{jk}')$ and $\mathbf{c}_j''=(c_{j1}'',\dots,c_{jk}''),$ where $c_{jl}'=f_l'\mathbf{b}_{lh_j}$ and $c_{jl}''=f_l''\mathbf{b}_{lh_j'},$ for some $f_l'\in H_j$ and $f_l''\in H_j',$ for all $l=1,2,\dots,k$ and $j=1,2,\dots,t.$
\end{itemize}
\end{theorem}

\begin{proof}
Apply Proposition~\ref{imagemiu2}, equation~\ref{decomp3}, and Chinese remainder algorithm (see \cite[Algorithm 5.4]{gathen}).
\end{proof}

\begin{theorem}\label{sigmaselfdual}
A code $C$ is a Euclidean self-dual $\sigma$-code of length $n$ over $\mathbb{F}_q$ if and only if
\[\mu(C)=\left(\bigoplus_{i=1}^sC_i\right)\oplus\left(\bigoplus_{j=1}^t(C_j'\oplus (C_j')^\bot)\right),\]
where
\begin{itemize}
\item[(i)] For any $\mathbf{c}_i\in C_i\leq G_i^k,$ $\mathbf{c}_i=(c_{i1},\dots,c_{ik}),$ where $c_{il}=f_l\mathbf{b}_{lg_i},$ for some $f_l\in G_i,$ for all $l=1,2,\dots,k$ and $i=1,2,\dots,s,$ and

\item[(ii)] For any $\mathbf{c}_j'\in C_j'\leq H_j^k,$  $\mathbf{c}_j'=(c_{j1}',\dots,c_{jk}'),$ where $c_{jl}'=f_l'\mathbf{b}_{lh_j},$ for some $f_l'\in H_j,$ for all $l=1,2,\dots,k$ and $j=1,2,\dots,t.$

\item[(iii)] Submodule $C_i$ is Hermitian self-dual over $G_i,$ for all $i=1,2,\dots,k,$ and

\item[(iv)] Submodule $(C_j')^\bot$ is the Euclidean dual of $C_j',$ for all $j=1,2,\dots,k.$
\end{itemize}
\end{theorem}

\begin{proof}
Apply Theorem~\ref{sigmadecomp} and \cite[Theorem 4.2]{ling2}.
\end{proof}

\newpage
\begin{table}[h]
	\centering
	\begin{tabular}{|c c c c c|}
		\hline\hline
		$n$ & $\sigma$ & Generator & $k$ & $d$ \\ [0.5ex]
		\hline\hline
		5 & $(1\;2)(3\;4)(5)$ & $(1,1,1)$ & $2$ & $3$\\
		5 & $(1\;2\;3)(4\;5)$ & $(1+Y^2,1+Y)$ & 3 & 2\\
		5 & $(1\;2\;3\;4(5)$ & $(1+Y^2+Y^3,1)$ & 4 & 2\\
		6 & $(1\;2)(3\;4)(5\;6)$ & $(1,Y,1+Y)$ & 2 & 4\\
		6 & $(1\;2\;3)(4\;5\;6)$ & $(1,1+Y^2)$ & 3 & 3\\
        6 & $(1\;2\;3\;4)(5\;6)$ & $(1,1+Y)$ & 4 & 2\\
        6 & $(1\;2\;3\;4\;5)(6)$ & $(1,1)$ & 5 & 2\\
		7 & $(1\;2\;3\;4)(5\;6\;7)$ & $(1+Y+Y^2,1)$ & 6 & 2\\
		7 & $(1\;2\;3\;4)(5\;6\;7)$ & $(1+Y,1+Y)$ & 5 & 2\\
        7 & $(1\;2\;3\;4)(5\;6)(7)$ & $(1+Y,1,1)$ & 3 & 4\\
		7 & $(1\;2\;3\;4\;5\;6)(7)$ & $(1+Y^2+Y^4,1)$ & 2 & 4\\
		8 & $(1\;2)(3\;4)(5\;6)(7\;8)$ & $(Y,1,1+Y,1)$ & 2 & 5\\
        8 & $(1\;2)(3\;4\;5)(6\;7\;8)$ & $(1+Y,1,1)$ & 3 & 4\\
        8 & $(1\;2\;3\;4\;5)(6\;7\;8)$ & $(1,1+Y+Y^2)$ & 5 & 2\\
		8 & $(1\;2\;3\;4\;5)(6\;7\;8)$ & $(1+Y+Y^2+Y^3,1+Y)$ & 6 & 2\\
        8 & $(1\;2\;3\;4\;5\;6\;7)(8)$ & $(1+Y^5+Y^6,1)$ & 7 & 2\\
		9 & $(1\;2)(3\;4)(5\;6)(7\;8)(9)$ & $(1,Y,1+Y,Y,1)$ & 2 & 6\\
        9 & $(1\;2\;3)(4\;5\;6)(7\;8\;9)$ & $(1,1+Y+Y^2,Y+Y^2)$ & 3 & 4\\
        9 & $(1\;2\;3\;4)(5\;6\;7\;8)(9)$ & $(1+Y,1+Y+Y^2,1)$ & 4 & 4\\
		9 & $(1\;2)(3\;4\;5\;6\;7)(8\;9)$ & $(1,1+Y+Y^3,1+Y)$ & 6 & 2\\
        10 & $(1\;2)(3\;4)(5\;6)(7\;8)(9\;10)$ & $(1+Y,Y,1,Y,1+Y)$ & 2 & 6\\
		10 & $(1\;2\;3)(4\;5\;6)(7\;8\;9)(10)$ & $(1,1+Y,Y^2,1)$ & 3 & 5\\
        10 & $(1\;2\;3\;4\;5)(6\;7\;8\;9\;10)$ & $(1+Y+Y^3+Y^4,1+Y^2+Y^3+Y^4)$ & 4 & 4\\
        10 & $(1\;2\;3\;4\;5)(6\;7\;8\;9\;10)$ & $(1+Y^3+Y^4,Y^3+Y^4)$ & 5 & 4\\
        11 & $(1\;2)(3\;4)(5\;6)(7\;8)(9\;10)(11)$ & $(1,Y,1+Y,1,1,1)$ & 2 & 7\\
        11 & $(1\;2\;3)(4\;5\;6)(7\;8\;9)(10\;11)$ & $(1,1+Y^2,1,1+Y)$ & 3 & 6\\
        11 & $(1\;2\;3\;4\;5)(6\;7\;8\;9\;10)(11)$ & $(1+Y,1+Y^3+Y^4,1)$ & 5 & 4\\
        12 & $(1\;2)(3\;4)(5\;6)(7\;8)(9\;10)(11\;12)$ & $(1,Y,1+Y,Y,1,1+Y)$ & 2 & 8\\
        12 & $(1\;2\;3)(4\;5\;6)(7\;8\;9)(10\;11\;12)$ & $(1,1+Y,Y,1+Y+Y^2)$ & 3 & 6\\
        12 & $(1\;2\;3\;4\;5)(6\;7\;8\;9\;10)(11\;12)$ & $(1,1+Y,1)$ & 5 & 4\\
        12 & $(1\;2\;3\;4\;5\;6)(7\;8\;9\;10\;11\;12)$ & $(1,1+Y^2+Y^3+Y^4+Y^5)$ & 6 & 4\\
        13 & $(1\;2)(3\;4)(5\;6)(7\;8)(9\;10)(11\;12)(13)$ & $(1,1,1,1+Y,Y,Y,1)$ & 2 & 8\\
        13 & $(1\;2\;3)(4\;5\;6)(7\;8\;9)(10\;11\;12)(13)$ & $(1+Y,Y,1,1+Y,1)$ & 3 & 7\\
        13 & $(1\;2\;3\;4)(\;5\;6\;7\;8)(9\;10\;11\;12)(13)$ & $(1+Y+Y^2,1,Y+Y^2+Y^3,1)$ & 4 & 6\\
        13 & $(1\;2\;3\;4\;5\;6)(7\;8\;9\;10\;11\;12)(13)$ & $(1,1+Y+Y^2+Y^4+Y^5,1)$ & 6 & 4\\
        14 & $(1\;2)(3\;4)(5\;6)(7\;8)(9\;10)(11\;12)(13\;14)$ & $(1,Y,Y,1+Y,1,1+Y,1)$ & 2 & 9\\
        14 & $(1\;2\;3)(4\;5\;6)(7\;8\;9)(10\;11\;12)(13\;14)$ & $(1,Y,1+Y,Y+Y^2,1+Y)$ & 3 & 8\\
        14 & $(1\;2\;3\;4)(5\;6\;7\;8)(9\;10\;11\;12)(13\;14)$ & $(1+Y,1+Y^2+Y^3,Y^2,1)$ & 4 & 7\\
        14 & $(1\;2\;\cdots\;7)(8\;9\;\cdots\;14)$ & $\begin{array}{c}
        (1+Y+Y^3+Y^4+Y^5+Y^6,\\
        1+Y+Y^2+Y^5+Y^6)
        \end{array}$ & 7 & 4\\
        15 & $(1\;2)(3\;4)\cdots(13\;14)(15)$ & $(1,Y,1+Y,Y,Y,1,1+Y,1)$ & 2 & 10 \\
        15 & $(1\;2\;3)\cdots (10\;11\;12)(13\;14)(15)$ & $(1,1+Y^2,Y,1+Y,1+Y,0)$ & 3 & 8\\
        15 & $(1\;\cdots\;4)\cdots(8\;\cdots\;12)(13\;14)(15)$ & $(1+Y+Y^3,Y,Y+Y^2,1,1)$ & 4 & 8\\

		\hline \hline
	\end{tabular}
	\caption{Binary $\sigma$-codes}
	\label{binary}
\end{table}
\newpage
\begin{table}[p]
	\centering
	\begin{tabular}{|c c c c c|}
		\hline\hline
		$n$ & $\sigma$ & Generator & $k$ & $d$ \\ [0.5ex]
		\hline\hline
		5 & $(1\;5)(2\;4)(3)$ & $(2Y,Y,2)$ & 2 & 3\\
        5 & $(1\;2)(3\;4\;5)$ & $(1+2Y,1+2Y^2)$ & 3 & 2\\
		5 & $(1\;2)(3\;4\;5)$ & $(1,2+2Y+Y^2)$ & 4 & 2\\
		6 & $(1\;2)(3\;4)(5\;6)$ & $(1+2Y,4+3Y,1+3Y)$ & 2 & 4\\
        6 & $(1\;2\;4)(3\;5\;6)$ & $(Y+Y^2,2Y)$ & 3 & 3\\
        6 & $(1\;2)(3\;4\;5\;6)$ & $(1,2+Y+2Y^2)$ & 4 & 2\\
		6 & $(1\;2\;3\;6\;5)(4)$ & $(2Y,1)$ & 5 & 2\\

		7 & $(1\;2)(3\;4)(5\;6)(7)$ & $(1+2Y,4+3Y,1+3Y,2)$ & 2 & 5\\
        7 & $(1\;2\;3)(4\;5\;6)(7)$ & $(1+Y,1+2Y+Y^2,1)$ & 3 & 4\\
        7 & $(1\;2\;3)(4\;5\;6\;7)$ & $(2+2Y+Y^2,2+2Y)$ & 5 & 2\\
		7 & $(1\;2\;3)(4\;5\;6\;7)$ & $(2+2Y+Y^2,1+2Y+2Y^2)$ & 6 & 2\\

        8 & $(1\;2)(3\;4)(5\;6)(7\;8)$ & $(1+2Y,1,1,1+Y)$ & 2 & 6\\
        8 & $(1\;\cdots\;4)(5\;\cdots\;8)$ & $(1+2Y+2Y^3,1+Y+2Y^2+Y^3)$ & 4 & 4\\
        8 & $(1\;2\;3\;4\;5\;6)(7\;8)$ & $(2+2Y+2Y^3+Y^5,Y)$ & 5 & 3\\
        8 & $(1\;2\;3\;4\;5\;6)(7\;8)$ & $(2+2Y+2Y^3+Y^4+Y^5,Y)$ & 6 & 2\\

        9 & $(1\;2)(3\;4)(5\;6)(7\;8)(9)$ & $(1+2Y,1,1+Y,1+Y,1)$ & 2 & 6\\
        9 & $(1\;2\;3\;4\;5\;6)(7\;8\;9)$ & $(1+2Y+2Y^2+Y^4+2Y^5,1+2Y+2Y^2)$ & 5 & 4\\
        9 & $(1\;2\;3\;4\;5\;6\;7)(8\;9)$ & $(Y+2Y^2+Y^4+2Y^5,1+Y)$ & 7 & 2\\

        10 & $(1\;2)(3\;4)(5\;6)(7\;8)(9\;10)$ & $(1+2Y,1,1+Y,1,1)$ & 2 & 7\\
        10 & $(1\;2\;3)(4\;5\;6)(7\;8\;9)(10)$ & $(1+2Y+Y^2,Y+2Y^2,Y+Y^2,1)$ & 3 & 6\\

        10 & $(1\;2\cdots7\;8)(9\;10)$ & $\begin{array}{c}
            \left(1+2Y+2Y^2+2Y^3+2Y^4\right.\\
            \left.+2Y^5+2Y^6+2Y^7,1+Y\right)
        \end{array}$ & 8 & 2\\

        11 & $(1\;2)\cdots(9\;10)(11)$ & $(1+Y,Y,Y,2+Y,1,1)$ & 2 & 8\\
        11 & $(1\;\cdots\;4)(5\;\cdots\;8)(9\;10)(11)$ & $(1+Y+Y^3,1+Y,Y,1)$ & 4 & 6\\

        11 & $(1\;2)(3\;4\cdots10\;11)$ & $\begin{array}{c}
        (2+2Y,2+Y+2Y^2+2Y^3+2Y^4\\
        +2Y^5+2Y^6+Y^7)
        \end{array}$ & 9 & 2\\

        12 & $(1\;2\;3)\cdots(10\;11\;12)$ & $(1+2Y^2,1+Y^2,1,1+2Y+Y^2)$ & 3 & 8\\
        12 & $(1\;2\;3\;4)\cdots(9\;10\;11\;12)$ & $(1+2Y^2+Y^3,Y+Y^2,Y+Y^2+Y^3)$ & 4 & 6\\
        12 & $(1\;\cdots\;10)(11\;12)$ & $\begin{array}{c}(1+Y+Y^2+Y^3\\
            +Y^4+Y^5+Y^6+Y^7+Y^8,1)
            \end{array}$ & 10 & 2\\

        13 & $(1\;2)\cdots(11\;12)(13)$ & $(1,1+2Y,Y,1+Y,1,1,1)$ & 2 & 9\\
        13 & $(1\;2\;3)\cdots(10\;11\;12)(13)$ & $(1+Y^2,2+Y+Y^2,2,2+Y,1)$ & 3 & 9\\
        13 & $(1\;\cdots\;4)\cdots(9\;\cdots\;12)(13)$ & $\begin{array}{c}(1+Y+Y^2+2Y^3,\\
                2+Y+Y^2+2Y^3,2Y+Y^2,1)\end{array}$ & 4 & 7\\
        13 & $(1\;\cdots\;5)(6\;\cdots\;10)(11\;12)(13)$ & $\begin{array}{c}
                (1+Y+Y^2+Y^4,\\
                1+2Y+2Y^4,1+Y,0)
              \end{array}$ & 5 & 6\\
        14 & $(1\;2)\cdots(13\;14)$ & $(1,1+2Y,Y,2+Y,1,1,1+Y)$ & 2 & 10\\
        14 & $(1\;2\;3)\cdots(10\;11\;12)(13\;14)$ & $(1+Y^2,2+Y^2,2+Y+Y^2,Y,1+y)$ & 3 & 9\\
        14 & $(1\;\cdots\;6)(7\;\cdots\;12)(13 14)$ & $\begin{array}{c}
           (1+Y^3+Y^5,\\
           1+Y^2+Y^4+Y^5,1)
          \end{array}$ & 6 & 6\\
        15 & $(1\;2)\cdots(13\;14)(15)$ & $\begin{array}{c}
           (1,1+2Y,1+Y,Y,1+2y,\\
           2,1,1+Y)
          \end{array}$ & 2 & 11\\

        15 & $(1\;2\;3)\cdots(13\;14\;15)$ & $\begin{array}{c}
            (1+Y^2,1+2Y+2Y^2,2+Y+Y^2,\\
            Y^2,Y+Y^2)
           \end{array}$ & 3 & 9\\

        15 & $(1\;\cdots\;7)(8\;\cdots\;14)(15)$ & $\begin{array}{c}
        (1+Y+Y^2+Y^3+Y^4+Y^6,\\
        Y^2+2Y^3+Y^4+Y^5,1)
        \end{array}$ & 7 & 6 \\
		\hline \hline
	\end{tabular}
	\caption{Ternary $\sigma$-codes}
	\label{ternary}
\end{table}
\newpage
\begin{table}[t]
	\centering
	\begin{tabular}{|c c c c c|}
		\hline\hline
		$n$ & $\sigma$ & Generator & $k$ & $d$ \\ [0.5ex]
		\hline\hline
        5 & $(1\;2)(3\;4)(5)$ & $(1,1+2Y,1)$ & 2 & 4\\
        5 & $(1\;2)(3\;4\;5)$ & $(3+Y,2+Y)$ & $4$ & $2$\\

        6 & $(1\;2\;3)(4\;5\;6)$ & $(1+3Y+2Y^2,1+2Y+3Y^2)$ & 3 & 4\\
        6 & $(1\;\cdots\;5)(6)$ & $(1+2Y+Y^2+3Y^3+Y^4,1)$ & 5 & 2\\

        7 & $(1\;2)(3\;4)(5\;6)(7)$ & $(1+2Y,3+4Y,1+4Y,1)$ & 2 & 5\\
        7 & $(1\;2\;3)(4\;5\;6)(7)$ & $(1+2Y+Y^2,3+Y,1)$ & 3 & 4\\
        7 & $(1\;2)(3\;4\;5\;6\;7)$ & $(1+Y,1+Y+Y^2+Y^4)$ & 5 & 2\\
        7 & $(1\;2\;3\;4\;5\;6)(7)$ & $(1+Y+Y^2+Y^3+Y4+2Y^5,1)$ & 6 & 2\\

        8 & $(1\;2)\cdots(7\;8)$ & $(1,Y,1+2Y,2+Y)$ & 2 & 6\\
        8 & $(1\;\cdots\;4)(5\;\cdots\;8)$ & $(1+2Y+3Y^2+4Y^3,4+Y+Y^2+2Y^3)$ & 4 & 4\\

        8 & $(1\;2)(3\;4\;5\;6\;7\;8)$ & $(1+Y,1+3Y+4Y^2+2Y^3+4Y^4)$ & 6 & 2\\
        8 & $(1\;\cdots\;7)(8)$ & $(1,1)$ & 7 & 2\\

        9 & $(1\;2)(3\;4)(5\;6)(7\;8)(9)$ & $(1+4Y,2+Y,3+4Y,3Y,1)$ & 2 & 7\\
        9 & $(1\;2\;3)(4\;5\;6)(7\;8\;9)$ & $(1+4Y,2+Y,1+2Y+3Y^2)$ & 3 & 6\\
        9 & $(1\;\cdots\;4)(5\;\cdots\;8)(9)$ & $(1+Y+2Y^2,1+Y^2+Y^3,1)$ & 4 & 5\\
        9 & $(1\;\cdots\;7)(8)(9)$ & $(1+Y^3+Y^4+Y^6,0,1)$ & 7 & 2 \\

        10 & $(1\;2)(3\;4)(5\;6)(7\;8)(9\;10)$ & $(1+4Y,2+Y,3+4Y,3Y,1+Y)$ & 2 & 8\\
        10 & $(1\;2\;3)(4\;5\;6)(7\;8\;9)(10)$ & $(1+2Y+Y^2,3+4Y,1+2Y+3Y^2,1)$ & 3 & 7\\
        10 & $(1\;\cdots\;5)(6\;\cdots\;10)$ & $(1+Y+Y^3,1+3Y+Y^2+Y^3)$ & 5 & 5\\
        10 & $(1\;\cdots\;8)(9\;10)$ & $(1,1)$ & 8 & 2\\

        12 & $(1\;2\;3)\cdots(10\;11\;12)$ & $(1+Y^2,1+Y+2Y^2,2+Y+2Y^2,1+3Y+4Y^2)$ & 3 & 8\\
        12 & $(1\;\cdots\;6)(7\;\cdots\;12)$ & $\begin{array}{c}
        (1+2Y+Y^2+3Y^3+Y^4+Y^5,\\
        1+Y+Y^3+Y^4+4Y^5)
        \end{array}$ & 6 & 6\\

        13 & $(1\;2)\cdots(11\;12)(13)$ & $\begin{array}{c}
        (1,1+2Y,1+3Y,\\
        Y,1+4Y,1+2Y,1)
        \end{array}$ & 2 & 10\\

        13 & $(1\;2\;3)\cdots(10\;11\;12)(13)$ & $\begin{array}{c}
        (1+Y^2,1+Y+3Y^2,\\
        1+Y+2Y^2,3+Y+2Y^2,1)
        \end{array}$ & 3 & 9\\

        13 & $(1\;2\;3\;4)\cdots(9\;10\;11\;12)(13)$ & $\begin{array}{c}
        (1+Y^2+Y^3,1+3Y+2Y^2+Y^3,\\
        2+3Y+Y^2+2Y^3,1)
        \end{array}$ & 4 & 8\\

        13 & $(1\;\cdots\;6)(7\;\cdots\;12)(13)$ & $\begin{array}{c}
        (1+2Y+3Y^2+4Y^3+Y^4+Y^5,\\
        1+Y+2Y^2+3Y^3+Y^4+Y^5,1)
        \end{array}$ & 6 & 6\\

		\hline \hline
	\end{tabular}
	\caption{5-ary $\sigma$-codes}
	\label{5ary}
\end{table}
\vspace{0.25cm}


\begin{thebibliography}{99}

%
%
	
	











%
%



\bibitem{gathen}
J. von zur Gathen and J. Gerhard, {Modern Computer algebra}, third edition, 2013, Cambridge University Press.

\bibitem{ling}
S. Ling and P. Sol$\acute{e}$, {\it On the algebraic structures of quasi-cyclic codes I: Finite fields}, IEEE Trans. Inf. Theory {\bf 47} (2001), no. 7, 2751--2760.

\bibitem{ling2}
S. Ling and P. Sol$\acute{e}$, {\it On the algebraic structures of quasi-cyclic codes II: Chain rings}, Des. Codes. Crypt. {\bf 30} (2003), 113--130.

\bibitem{roman}
S. Roman, {Advanced linear algebra}, third edition, 2008, Springer.

\bibitem{pless}
V. Pless and W.C. Huffman(editor), {\it Handbook of coding theory}, 1998, Elsevier.

\end{thebibliography}
\end{document}